\documentclass[%
 reprint,
superscriptaddress,
 amsmath,amssymb,
 aps,
longbibliography
]{revtex4-1}
\usepackage{graphicx}
\usepackage{dcolumn}
\usepackage{bm}

\usepackage[utf8]{inputenc}
\usepackage{graphicx}
\usepackage{amsmath}
\usepackage{amssymb}
\usepackage{fullpage}
\usepackage{xcolor,array,bm}
\usepackage{amsthm}
\usepackage{braket}
\usepackage{tikz-cd}
\usepackage{nicefrac,xfrac}
\usepackage{wrapfig}
\usepackage{caption}
\usepackage{subcaption}
\usepackage{physics}
\usepackage{bbold}
\usepackage{mathtools}

\RequirePackage{amsmath}
\RequirePackage{amsthm}
\RequirePackage{amssymb}
\RequirePackage{setspace}

\usepackage{xcolor}
\definecolor{darkblue}{rgb}{0,0,.5}
\usepackage[
            colorlinks=true,
            linkcolor=darkblue, urlcolor=darkblue, citecolor=darkblue,
            raiselinks=true,
            bookmarks=true,
            bookmarksopenlevel=1,
            bookmarksopen=true,
            bookmarksnumbered=true,
            hyperindex=true,
            plainpages=false,
            pdfpagelabels=true,
            pdfstartview=FitH,
            pdfstartpage=1,
            pdfpagelayout=OneColumn
            ]{hyperref}

\newcommand{\bigO}{$\mathcal{O}$}
\newcommand{\tens}{\otimes}

\newcommand{\hal}{\mathcal{H}}
\newcommand{\del}{\partial}

\newcommand{\bigslant}[2]{{\raisebox{.2em}{#1}\left/\raisebox{-.2em}{#2}\right.}}
\newcommand{\homk}{\textbf{$H_k$}}

\newtheorem{theorem}{Theorem}

\newtheorem*{problem*}{Problem}

\begin{document}

\preprint{APS/123-QED}

\title{Complexity-Theoretic Limitations on Quantum Algorithms for Topological Data Analysis}

\author{Alexander Schmidhuber}
\email{alexsc@mit.edu}
 \affiliation{Institute for Theoretical Physics, ETH Zurich, Zurich 8092, Switzerland}
 \affiliation{Center for Theoretical Physics, Massachusetts Institute of Technology, Cambridge, MA 02139}
\author{Seth Lloyd}%
 \affiliation{Department of Mechanical Engineering, Massachusetts Institute of Technology, Cambridge, MA 02139}
\affiliation{Turing Inc., Cambridge, MA 02139}

\date{December 21, 2023}

\begin{abstract}
Quantum algorithms for topological data analysis (TDA) seem to provide an exponential advantage over the best classical approach while remaining immune to dequantization procedures and the data-loading problem. In this paper, we give complexity-theoretic evidence that the central task of TDA -- estimating Betti numbers -- is intractable even for quantum computers. Specifically, we prove that the problem of computing Betti numbers exactly is \#P-hard, while the problem of approximating Betti numbers up to multiplicative error is NP-hard. Moreover, both problems retain their hardness if restricted to the regime where quantum algorithms for TDA perform best. Because quantum computers are not expected to solve \#P-hard or NP-hard problems in subexponential time, our results imply that quantum algorithms for TDA offer only a polynomial advantage in the worst case. We support our claim by showing that the seminal quantum algorithm for TDA developed by Lloyd, Garnerone and Zanardi achieves a quadratic speedup over the best known classical approach in asymptotically almost all cases. Finally, we argue that an exponential quantum advantage can be recovered if the input data is given as a specification of simplices rather than as a list of vertices and edges.
\end{abstract}

\maketitle


\section{Introduction}
One of the major challenges in quantum information science is the development of quantum algorithms with an exponential advantage over the best classical counterpart. Some of the most promising candidates in that regard have been quantum machine learning (QML) algorithms \cite{QML1,Harrow_2009,QML2,QML3}. However, two issues have challenged the prospect of exponential speedups for many QML algorithms in recent years. On the one hand, a series of impressive dequantization results by Tang et al. \cite{Tang,Chia_2020} showed that several QML problems based on low-rank matrix analysis can be solved efficiently on a classical computer. On the other hand, most remaining QML algorithms require the encoding of a large input into a quantum state \cite{scott}, which requires large quantum random access memory (qRAM).
 
These issues have resulted in increased efforts to find quantum algorithms that are immune to dequantization procedures and do not suffer from the data-loading problem. One prominent example is the quantum algorithm for topological data analysis (TDA) developed by Lloyd, Garnerone and Zanardi \cite{LloydAlgo}, which has experienced a recent surge of attention from a theoretical \cite{linearDepth,logDepth,persistent,quantumAdvantage,reviewBetti,GoogleBettiBerry,AmazonBettiMcArdle} as well as experimental \cite{experimentalBetti,IBMresult} perspective. This quantum algorithm, which we will refer to as LGZ-algorithm, does not require a large classical input and is immune to standard dequantization procedures. Nevertheless, it solves the task of estimating the Betti numbers of a simplicial complex - the essential task of TDA - and does so seemingly exponentially faster than the best classical algorithm, provided that certain assumptions regarding the input hold. 

In this paper, we show that any quantum algorithm for TDA runs in exponential time in the worst case by proving that (under widely believed complexity-theoretic assumptions) the problem of topological data analysis is intractable even for quantum computers. Specifically, we prove that the problem of computing Betti numbers exactly is \#P-hard and that the problem of approximating Betti numbers up to multiplicative error is NP-hard. Moreover, we show that both problems retain their hardness if the input is restricted to clique-dense complexes, which is the regime where the LGZ-algorithm works best \cite{LloydAlgo,quantumAdvantage,reviewBetti}. Building upon this result, we argue that the LGZ-algorithm for TDA runs in exponential time not only in the worst case, but for asymptotically almost all inputs. We provide direct evidence for this claim by investigating the runtime of the LGZ-algorithm on random Vietoris-Rips complexes, for which we show that it achieves a Grover-like speedup over the best known classical algorithm in almost all cases. Our complexity-theoretic results and average-case analysis do not rule out an exponential quantum advantage over the best classical algorithm for some specific inputs. Indeed, there exists a known family of simplicial complexes for which the LGZ-algorithm achieves a super-polynomial speedup \cite{GoogleBettiBerry}. 


We further provide a complexity-theoretic analysis of each step of the LGZ-algorithm, through which we show that the subroutine limiting quantum advantage in TDA is not actually the computation of Betti numbers given a description of the simplicial complex, but the construction of a quantum state representing the simplicial complex from a description of the underlying graph. We show that the latter problem is \#P-hard by itself. It is thus the problem of sampling simplices given a list of vertices and edges that forms the bottleneck of topological data analysis.   Provided with an oracle that enables random sampling from the $k$-simplices of a simplicial complex, a modification of the LGZ-algorithm can estimate its Betti numbers in polynomial time, exponentially faster than the best known classical approach with access to the same oracle. We study problems in data analysis and computational topology where such oracles appear naturally. In particular, when the data is given in the form of a list or specification of sets (simplices) and their members (vertices), an exponential quantum advantage could be recovered. This is the case, for example, when we are given the ability to sample from a list of Facebook groups together with the members of the sampled groups, and when a $k$-simplex is defined to be a group of $k+1$ individuals all of whom are members of the same group. 

The remaining article is structured as follows: In section 2, we introduce the framework of topological data analysis and describe the original LGZ Betti number quantum algorithm developed by Lloyd, Garnerone and Zanardi, along with a brief summary of recent improvements. In section 3, we study the following formal problem of computing Betti numbers:
\begin{problem*}[\textbf{Betti}]
Given a clique complex $S$ defined by its vertices and edges and an integer $k\geq0$ as input, output the $k$-th Betti number $\beta_k$ of $S$.
\label{prob:Betti} 
\end{problem*}
\noindent The problem is defined exactly in the way that it would appear in practical applications of TDA, where the vertices represent a set of data points embedded in a metric space, and the edges describe whether two points are close to each other or not. The complexity of computing Betti numbers was left as an open problem in the review article \cite{openProblem}. We show that the problem $\textbf{Betti}$ is intractable for quantum computers
even if restricted to the optimal regime of clique-dense complexes, by establishing its hardness via the following theorem:
\begin{theorem}The problem \emph{\textbf{Betti}} is \emph{(a)} \#P-hard and \emph{(b)} remains \#P-hard when the input is restricted to clique-dense complexes.\label{thm:Betti}\end{theorem} 
\noindent  However, the LGZ-algorithm does not compute Betti numbers exactly but only approximately. Theorem \ref{thm:Betti} is therefore not enough to reason about the runtime of the LGZ-algorithm. Indeed, there are examples of polynomial quantum algorithms that efficiently approximate \#P-hard problems, such as the Jones polynomial \cite{Jones} or the Potts model \cite{Potts}. In section 4, we show that even approximating Betti numbers up to any multiplicative error remains intractable for quantum computers (again, under the widely-held complexity-theoretic assumption that quantum computers cannot access NP-hard problems). To do so, we will study the problem of deciding which Betti numbers are non-zero. 
\begin{problem*}[\textbf{Homology}]
Given a clique complex $S$ defined by its vertices and edges and an integer $k\geq0$ as input, output \emph{true} if $\beta_k > 0$ and \emph{false} if $\beta_k = 0.$\label{prob:Homology}
\end{problem*}
\noindent Problem \textbf{Homology} is the natural decision version of the counting problem \textbf{Betti}. It was already shown by Adamaszek and Stacho \cite{adamaszek2016} that this problem is NP-hard for complexes defined on co-chordal graphs. Here, we strengthen this result by showing that it remains NP-hard on clique-dense complexes.
\begin{theorem}
The problem \emph{\textbf{Homology}} is \emph{(a)} NP-hard and \emph{(b)} remains NP-hard when the input is restricted to clique-dense complexes.\label{thm:Homology}
\end{theorem}
\noindent Evidently, Theorem \ref{thm:Homology} directly implies that computing any multiplicative approximation of Betti numbers is NP-hard and therefore not accessible to quantum computers. Because Theorem \ref{thm:Betti} and Theorem \ref{thm:Homology} hold even if restricted to the regime where the currently known quantum algorithms for TDA perform best - and because the relative size of this regime compared to all possible simplicial complexes tends towards $0$ as the number of vertices increases - our results imply that the LGZ-algorithm runs in exponential time for almost all cases. 
In section 5, we verify this claim by investigating the runtime of the LGZ-algorithm on random Vietoris-Rips complexes. We show that it achieves a Grover-like quadratic speedup over classical algorithms in asymptotically almost  cases. However, up to super-polynomial speedups can be achieved for specifically selected inputs.  

In section 6, we subsequently discuss problems beyond TDA for which a modification of the LGZ-algorithm achieves an exponential speedup. The main assumption enabling exponential quantum advantage is the access to random sampling from the simplices of a simplicial complex. We list examples of problems where such an oracle might appear naturally. Finally, in section 7, we comment on the computational hardness of computing normalized Betti numbers, which is a quantity that is not directly useful in practice but a more natural output of quantum algorithms for TDA.\\

\noindent\textbf{Remark:} \emph{While preparing for publication, simultaneous independent work by Crichigno and Kohler \cite{crichigno2022clique} appeared on the arXiv. The authors prove that the problem \textbf{Homology} is QMA1-hard by investigating supersymmetric quantum systems. This result strengthens Theorem 2 stated here and yields an alternative proof of Theorem 1. }
\section{Topological Data Analysis and the LGZ-Algorithm}

Topological data analysis is a recent approach to the analysis of large datasets that are high-dimensional, noisy or incomplete. The goal of TDA is to describe the shape of a dataset by extracting robust features - 
topological invariants - which are inherently insensitive to local noise. These topological invariants are the Betti numbers $\beta_k$, which for every $k$ count the number of $k$-dimensional holes in the dataset. While the theory behind TDA is based on algebraic topology, it can be understood through simple linear-algebraic terms. This section gives a concise description of the TDA pipeline and its two main procedures: Representing a dataset as a simplicical complex and subsequently computing its shape by estimating the Betti numbers of the complex.
\subsection{Topological Data Analysis}
\begin{figure*}
    \centering
    \includegraphics[width=2\columnwidth]{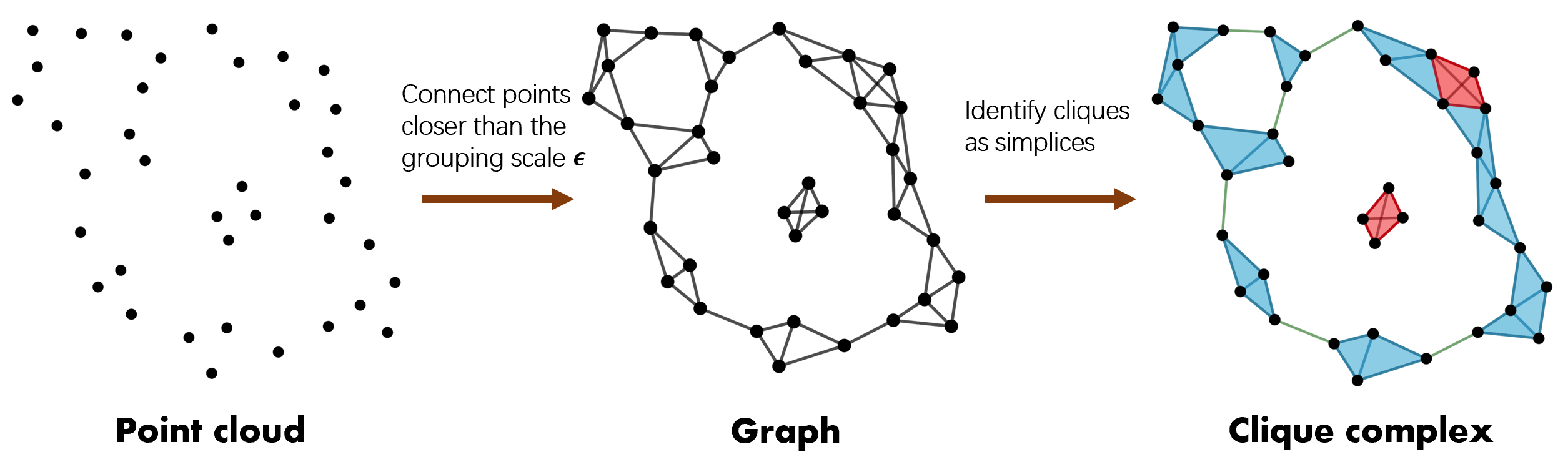}
     \caption{The construction of a simplicial complex from a point cloud. First, a grouping scale $\epsilon$ is chosen and every two data points with pairwise distance smaller than $\epsilon$ are connected. The resulting graph is subsequently promoted to a (simplicial) clique complex by identifying $k+1$-cliques with $k$-simplices.}
    \label{fig:pipeline}
\end{figure*}
The starting point of TDA is as follows: A dataset of interest is represented by a set of $n$ points $\{x_i\}_{i=1}^n$ (a point cloud) embedded in a metric space. For simplicity, we will study the example where $x_i$ are embedded in the real vector space $\mathbb{R}^d$ for some positive integer $d>0$, but the formalism extends to arbitrary manifolds. A point cloud is not connected and consequently has a trivial topology. To study the shape of the data set, the point cloud is thus first promoted to a topological object - a so-called simplicial complex - which can be interpreted as the generalization of a graph. This is done in two steps. First, a \textit{grouping scale} $\epsilon>0$ is chosen and every two points $x_i, x_j$ for which \begin{equation}
||x_i-x_j|| < \epsilon
\end{equation} are connected. The point cloud and resulting graph $G$ are depicted in Figure \ref{fig:pipeline}. 
Because a graph is an object with at most one-dimensional building blocks (edges), it does not yet allow the extraction of higher dimensional features. To remedy this, the graph $G$ is promoted to a simplicial clique complex by identifying $(k+1)$-cliques in $G$ with \emph{$k$-simplices.} A $k$-simplex $s_k$ of $G$ is thus a fully connected $k$-dimensional subgraph of $G$, i.e. a set \footnote{In algebraic topology, a simplex is defined as the linear span of an \emph{ordered} set. The ordering will be fixed later on by an ordering of the qubits encoding the simplices into quantum states. The particular ordering chosen does not impact topological invariants such as Betti numbers. We will thus omit it in this preliminary introduction to TDA, which also facilitates the description of our proofs in sections \ref{sec:counthard} and \ref{sec:approxhard}.} of $k+1$ points $\{v_0,...,v_k \} \subset \{x_i\}_{i=1}^n$ that are pairwise connected to each other. By this definition, a 0-simplex is a point, a 1-simplex a line, a 2-simplex a triangle, a 3-simplex a tetrahedron and so on. Simplices are visualized in Figure \ref{fig:boundary}. The set of all possible simplices that can be constructed on $G$ is called the \textit{simplicial complex} $S$ of the dataset. This particular type of simplicial complex -- constructed by identifying cliques of $G$ with simplices -- is also referred to as \emph{clique complex} of $G$, denoted $Cl(G)$. Furthermore, the clique complex of a graph obtained by connecting pairwise close points embedded in a metric space is called \emph{Vietoris-Rips} complex.

\noindent Note that $S$ is closed under inclusion: If a $k$-simplex $s_k = \{v_0,...,v_k\}$ is in the simplicial complex, the $(k-1)$-simplex $s_{k-1}(j)$ obtained by omitting the vertex $v_j$ from $s_k$ is also an element of the complex. These are the boundary simplices of $s_k$. Concretely, a line is bounded by two points, a triangle by three lines, and so on. 

In anticipation of an quantum algorithm for TDA, and to avoid repetition at a later stage, we will write simplices as quantum states of a $n$-qubit system. Because each $k$-simplex $s_k=\{v_0,...,v_k \}$ defines a subset of $\{x_i\}_{i=1}^n$, it can be naturally identified with a computational basis state $\ket{s_k}$ of Hamming weight $k+1$ on $n$ qubits, where the $j$-th bit of $\ket{s_k}$ is one if $s_k$ contains the vertex $x_j$, and zero else. If we denote by $\hal_k$ the Hilbert space spanned by the computational basis states of Hamming weight $k+1$, we have $\ket{s_k} \in \hal_k$ for all $k$ and $\oplus_{k=-1}^{n-1} \hal_k = \mathbb{C}^{2n}.$ It will become convenient later to further define the set $S_k$ of all $k$-simplices in the simplicial complex $S$, and the Hilbert space $\hal_k(S) \subset \hal_k$ spanned by $\{\ket{s_k} |\ s_k \in S_k\}$.
\subsection{Simplicial Homology} Having constructed the simplicial complex $S$ as a representation of the dataset $\{x_i\}_{i=1}^n$, TDA aims to describe the shape of $S$. From a topological standpoint, this shape is described by the \textit{Betti-numbers} $\beta_0,...\beta_n$ of $S$, which intuitively count the number of $k$-dimensional holes in the simplicial complex. We will first state their technical definition: Betti numbers are defined as the ranks of the homology groups induced by the boundary map $\del_k : \hal_k(S) \to \hal_{k-1}(S)$, which is a combinatorial operator
\begin{equation}
    \del_k\ket{s_k} = \sum_{i=0}^{k}(-1)^i\ket{s_{k-1}(i)},
\end{equation} mapping each $k$-simplex to the alternating sum of its boundary simplices. The action of $\del$ is visualized in Figure \ref{fig:boundary}.
\begin{figure*}
\centering
\begin{subfigure}{.5\textwidth}
  \centering
    \includegraphics[height=5cm]{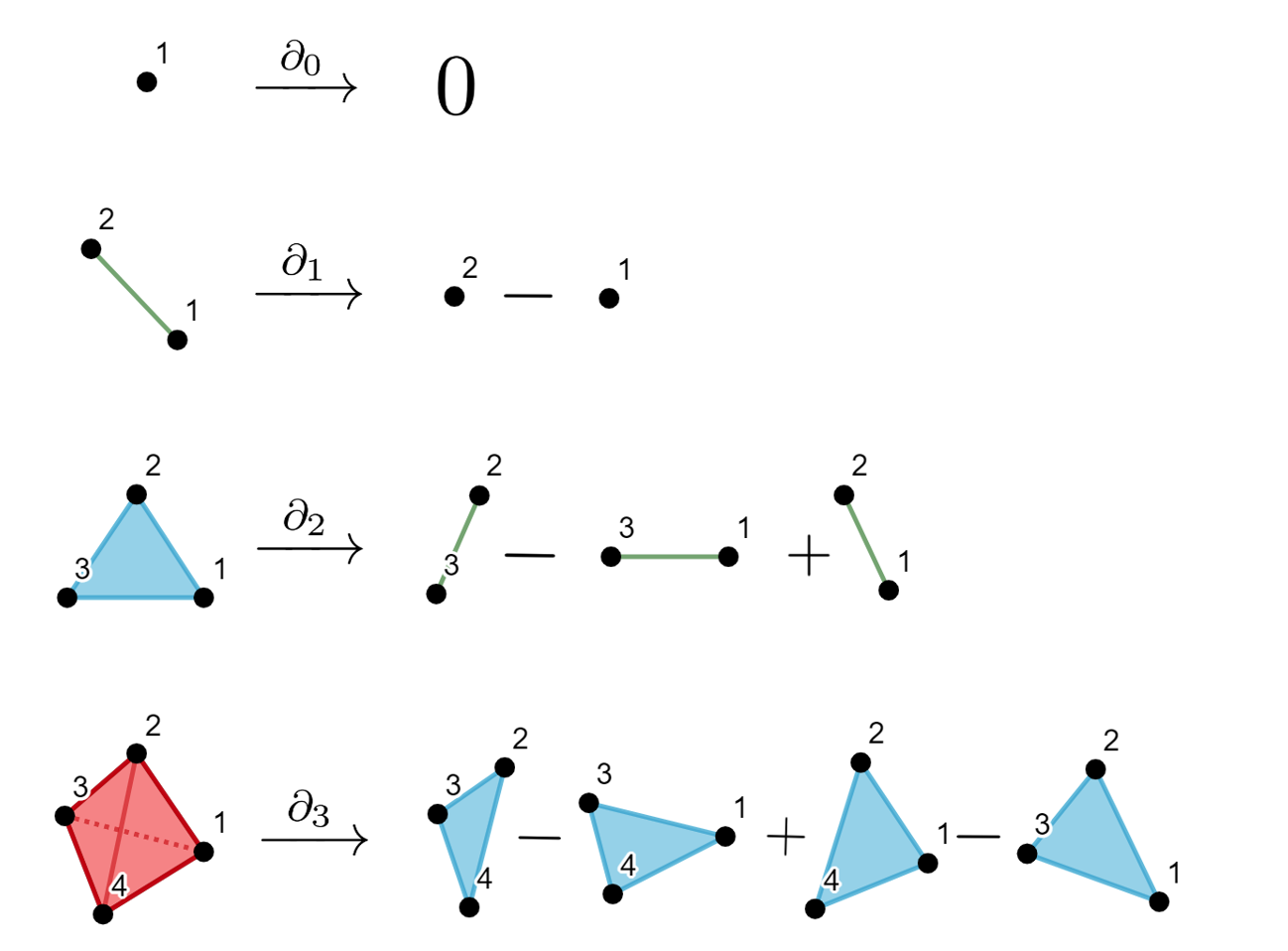}
    \caption{}
    \label{fig:boundary}
\end{subfigure}%
\begin{subfigure}{.5\textwidth}
  \centering
\includegraphics[height = 5cm]{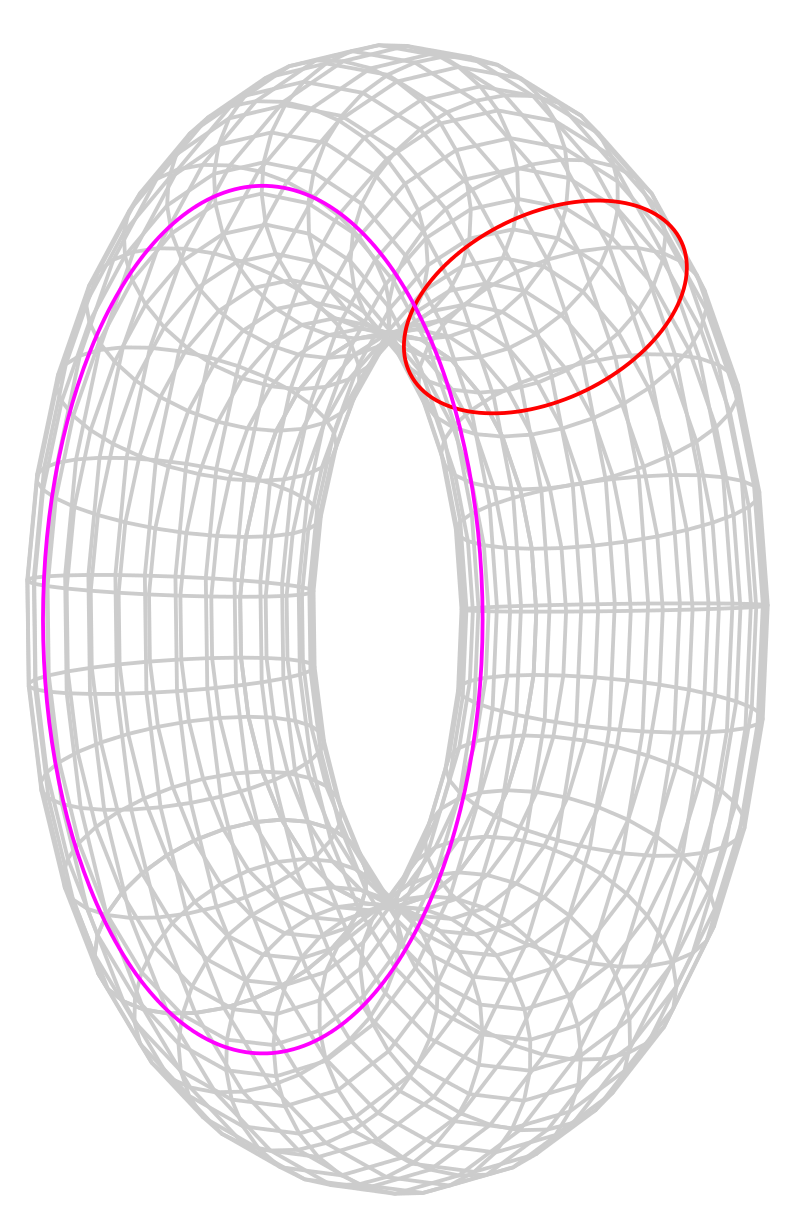}
    \caption{}
    \label{fig:torus}
\end{subfigure}
\caption{\textbf{(a)} A visualization of the simplicial boundary operator acting on a 0-simplex (point), 1-simplex (line), 2-simplex (triangle) and 3-simplex (tetrahedron). More generally, a $k$-simplex is mapped to the alternating sum of its $(k+1)$ boundary simplices. \textbf{(b)} A 2-dimensional torus embedded in $\mathbb{R}^3$. The torus has one connected component, two distinct holes (red and purple circle) and a single void (the interior). Its first three Betti numbers are thus $(\beta_0,\beta_1,\beta_2) = (1,2,1)$ and all higher orders vanish.}
\end{figure*}
As before, $\ket{s_{k-1}(i)}$ is the ($k-1$)-simplex obtained by setting the $i$-th entry of the bitstring $\ket{s_k}$ to zero. Note that the boundary of a boundary is zero,\begin{equation}
    \del_{k-1}\del_k = 0.
    \label{boundary}
\end{equation} Thus $\del_k$ induces a (finite) chain-complex on $S$, \begin{equation}
0 \xrightarrow{\hspace{2.5pt}\del_{n}\hspace{2.5pt} }  \hal_{n-1}(S) \xrightarrow{\del_{n-1}}\hdots  \xrightarrow{\hspace{4pt}\del_{1}\hspace{4pt} } \hal_0(S)  \xrightarrow{\hspace{3pt}\del_{0}\hspace{3pt} } 0,
\end{equation} with the corresponding $k$-th \textit{homology group} \textbf{$H_k$} defined as the quotient group of the kernel of $\del_k$ over the image of $\del_{k+1}$: \begin{equation}
    \textbf{$H_k$} = \bigslant{Ker $\del_k$}{Im $\del_{k+1}$}.
\end{equation} The $k$-th Betti number $\beta_k = \dim$ \homk \ equals the dimension of the $k$-th homology group \footnote{Because we have chosen to represent simplices as basis elements of a quantum mechanical Hilbert space, all homologies in this paper are computed over the field $\mathbb{C}$. In most applications of TDA, the data is however embedded in a \emph{real} metric space, hence one is interested in the homology over the reals $\mathbb{R}$. The simplicial Laplacian is real and symmetric in both cases, so this distinction does not matter.}.
Intuitively, the $k$-th Betti number counts the number of $k$-simplices in the complex that are boundaryless but not themselves a boundary, which can be interpreted as the number of $k$-dimensional holes. Thus $\beta_0$ counts the number of connected components, $\beta_1$ the number of holes, $\beta_2$ the number of voids, and so on. The Betti numbers of a torus are visualized in Figure \ref{fig:torus}. 
As for any chain complex that satisfies \eqref{boundary}, Hodge theory provides a relationship between the homology group $H_k$ and the \textit{Hodge Laplacian} $\Delta_k :=\del_{k}^*\del_{k} + \del_{k+1}\del_{k+1}^*$: 
\begin{equation}
    \homk \cong \text{Ker } \Delta_k = \text{Ker }(\del_{k}^*\del_{k} + \del_{k+1}\del_{k+1}^*).
    \label{Hodge}
\end{equation} In the special case presented here, this Hodge relation can be understood intuitively as follows: Both $\del_{k}^*\del_{k}$ and $\del_{k+1}\del_{k+1}^*$ are positive Hermitian operators, thus any state $\ket{\psi}$ annihilated by the Hodge Laplacian lies in the intersection of the kernels of $\del_k$ and $\del_{k+1}^*$. The first operator enforces that $\ket{\psi}$ has no boundary, and the second operator guarantees that only one representative per equivalence class is counted. Two cycles are considered equivalent if they can be continuously deformed to each other within the simplicial complex, i.e. if they differ by a boundary. Equation \eqref{Hodge} reduces the computation of $\beta_k$ to the linear-algebra problem of constructing and estimating the rank of a ${n\choose k+1}\times{n\choose k+1}$ matrix. In the next section, we describe a quantum algorithm for precisely this task, developed by Lloyd, Garnerone and Zanardi.
\subsection{The LGZ-Algorithm}
The LGZ quantum algorithm \cite{LloydAlgo} estimates the Betti numbers of a simplicial complex through repeated quantum phase estimation. Before providing a detailed description, we summarize its three main steps:
First, create the mixed density state
\begin{equation}
    \rho_k = \frac{1}{|{S_k}|} \sum_{s_k \in S_k} \ket{s_k}\bra{s_k}
\end{equation} of all $k$-simplices in the simplicial complex $S$. Here $|S_k|$ is the number of $k$-simplices of $S$. Secondly, exponentiate the Hodge Laplacian $\Delta$, \begin{equation}
  \Delta = \bigoplus_{k=0}^{n-1} \ \Delta_k = \begin{pmatrix}
\Delta_0 & 0 & \cdots& 0\\
0 & \Delta_1 & & \vdots \\
\vdots &  & \ddots & \vdots \\
0 & \cdots & \cdots & \Delta_{n-1}
\end{pmatrix}
\end{equation} and perform quantum phase estimation of $e^{i\Delta}$ with the eigenvector register starting in $\rho_k$. This yields random sampling from the eigenvalue distribution of $\Delta_k$. In a third step, steps one and two are repeated multiple times. The relative frequency of the zero-eigenvalue gives an estimate of the Kernel dimension of $\Delta_k$, and thus of the $k$-th Betti number: \begin{equation}
    \beta_k = \dim \homk = \dim \text{Ker } \Delta_k.
\end{equation}
In the following, each step is explained in detail. \\
\textbf{1. Projection onto $k$-simplices} \\
A variety of approaches can be used to arrive at the state $\rho_k$. LGZ construct it via Grover search applied to the membership function $f_k$, with \begin{equation}
    f_k(s) = \begin{cases} 
      1 & s \in S_k \\
      0 & \text{else} 
   \end{cases}.
\end{equation}
Given access to the pairwise distances between points $x_i$, $f_k$ can be implemented in \bigO$(k^2)$ steps. The multi-solution version of Grover's algorithm then allows us to construct \begin{equation}
    \ket{S_k} = \frac{1}{|{S_k}|^{1/2}} \sum_{s_k \in S_k} \ket{s_k},
\end{equation} while simultaneously revealing $|S_k|$, in \bigO$\left(\zeta_k^{-1/2}\right)$ calls to the function $f$. Here, \begin{equation}
\zeta_k = \frac{|S_k|}{{n \choose {k+1}}}    
\end{equation} is the fraction of occupied $k$-simplices of the simplicial complex $S$. The mixed state \begin{equation}
    \rho_k = \frac{1}{|{S_k}|} \sum_{s_k \in S_k} \ket{s_k}\bra{s_k}.
\end{equation} is then constructed by including an additional $n$-qubit ancilla system, applying a CNOT gate to each qubit in $\ket{S_k}$ and subsequently tracing out (ignoring) the ancilla. \\
\textbf{2. Quantum Phase Estimation} \\
The second step is to perform quantum phase estimation (QPE) with the operator $e^{i\Delta}$ on the state $\rho_k$. While directly exponentiating the Hodge Laplacian $\Delta$ might be costly in general, it is the square of a sparse operator $B$ defined as 
\begin{equation}
  B = \begin{pmatrix}
0 & \del_1 & 0 &\cdots& \cdots& 0\\
\del_1^\dagger & 0 & \del_2&  & & \vdots \\
0 & \del_2^\dagger & 0 & \ddots& &\vdots \\
 \vdots& & \ddots & \ddots & \ddots &\vdots \\
 \vdots& &  &\ddots & 0 & \del_{n-1} \\
0 &  \cdots& \cdots & \cdots& \del_{n-1}^\dagger & 0
\end{pmatrix}.\end{equation}
It is straightforward to verify that $B^*B = B^2 = \Delta$, and thus Ker $B =$ Ker $\Delta$. Because $B$ is the square root of a Laplacian, it is often referred to as \textit{Dirac operator}. Furthermore, because $\del_k$ is a ${n\choose k}\times{n\choose k+1}$ matrix with $n-k$ nonzero entries in each row and $k+1$ nonzero entries in each column, $B$ is $n$-sparse. Thus standard protocols allow the Hamiltonian simulation of $e^{iB}$ using \bigO$(n^3)$ gates. To get accurate results from QPE, the spectrum of $B$ has to be considered: Denote the smallest non-zero eigenvalue and the largest eigenvalue of $\Delta$ by $\lambda_{min}$ and $\lambda_{max}$, respectively. To avoid multiples of $2\pi$ in the exponent of $e^{iB}$, the Dirac operator has to be rescaled by $\lambda_{max}^{-1}$. Therefore the phase estimation has to be executed with precision at least $\kappa^{-1} = \frac{\lambda_{min}}{\lambda_{max}}$ to resolve whether a eigenvalue is zero or not, which yields a total costs of \bigO$(n^3 \kappa)$ for the QPE subroutine.  \\
\textbf{3. Kernel Estimation} \\
Repeating the above procedure $M$ times results in $M$ samples $\{\lambda_i\}_{i}$ from the eigenvalue distribution of $\Delta_k$. The relative frequency of the eigenvalue zero gives an additive estimate of the \emph{normalized} Betti number $c_k$,\begin{equation}
    c_k = \frac{\beta_k}{|S_k|} = \frac{\dim \homk}{|S_k|} \approx \frac{|\{i:\lambda_i = 0\}|}{M}.
\end{equation} 
The Hoeffding inequality shows that $M=\mathcal{O}(\delta^{-2})$ samples suffice to estimate $c_k$ up to additive error of $\varepsilon$. Combined with the state preparation and quantum phase estimation, the LGZ-algorithm computes an $\varepsilon$-additive approximation of the normalized $k$-th Betti number in time
\begin{equation}
\mathcal{O}\left(\frac{n^3\kappa  +n k^2\zeta_k^{-1/2}}{\varepsilon^{2}}\right).
\label{rt:norm}
\end{equation} 
This closely resembles the original runtime analysis of the LGZ-algorithm carried out in \cite{LloydAlgo}, however, there $\varepsilon$ was mistakenly labeled a multiplicative error. We will argue below that it is indeed preferable to state the runtime in terms of a multiplicative approximation error, as originally intended.

Note that this quantum algorithm not only estimates the $k$-th Betti number, but also supplies the corresponding eigenvectors which are the harmonic representatives of the homology \footnote{If the grouping scale $\epsilon$ varies reasonably slow during a filtration step, the harmonic representative of a persistent hole will not change by too much. Thus quantum phase estimation applied to the initial harmonic representative will have a high chance of projecting onto the new harmonic representative. This describes a quantum algorithm that can check whether a given hole persists.}. 
\subsection{Analysis}
The best known classical algorithm computes the $k$-th Betti number $\beta_k$ in time $\mathcal{O}\left(n \choose k+1\right)$, where $0\leq k \leq n-1$. At first glance, equation \eqref{rt:norm} thus seems to provide an exponential quantum advantage if $k$ scales with $n$. Two issues might reduce this advantage \cite{reviewBetti,quantumAdvantage}. The first concerns the fraction $\kappa= $$\frac{\lambda_{max}}{\lambda_{min}}.$ While $\lambda_{max}$ is bounded from above by the Gershgorin circle theorem \cite{reviewBetti}, there are no known lower bounds on the smallest eigenvalue $\lambda_{min}$ of the simplicial Laplacian $\Delta_k$. In some cases $\lambda_{min}$ is exponentially small, resulting in an exponential runtime of the LGZ-algorithm. Classes of simplicical Laplacians that have a minimum eigenvalue scaling at least inverse polynomially in the number of vertices have been identified in e.g. \cite{linearDepth,quantumAdvantage}. The second, more pressing issue concerns the fraction \begin{equation}
    \zeta_k = \frac{|S_k|}{{n \choose {k+1}}}.
\end{equation}  If the number of occupied $k$-simplices $|S_k|$ is small compared to ${n\choose k+1}$, the runtime of the quantum algorithm will again only provide a quadratic speedup over the best classical algorithm. A necessary condition for a polynomial runtime is thus that the simplicial complex is \emph{clique-dense} \cite{LloydAlgo,linearDepth,quantumAdvantage}, i.e. that  \begin{equation}
    \zeta_k^{-1} = {n\choose k+1}/ |S_k| \in \mathcal{O}(\text{poly}(n)).
    \label{cliqueDense}
\end{equation} Note that the restrictiveness of the requirement \eqref{cliqueDense} depends entirely on how $k$ scales with the number of vertices $n$. For constant $k$, every graph fulfills \eqref{cliqueDense} trivially. If $k$ grows slowly with $n$, such as $k = \Theta(\log(n)),$ the fraction $\zeta_k^{-1}$ will still not grow super-polynomially in the generic case, such as the random Erdös-Renyi model \cite{erdds1960random}. The regime where classical algorithms require exponential time, and  thus where an exponential quantum advantage is possible, is when $k$ grows asymptotically as $k = \Theta(n)$. This is also the regime where \eqref{cliqueDense} becomes restrictive, the regime originally considered in \cite{LloydAlgo}, and the regime where standard complexity-theoretic arguments apply. Throughout this paper, we will thus study the setting where $k = \Theta(n)$ scales linearly with $n$ (potentially up to logarithmic factors), unless specifically noted otherwise. In this case, the relative size of the clique-dense regime for random Vietoris-Rips complexes as defined in section \ref{sec:bettirandom} tends asymptotically towards 0.

Several classes of clique-dense complexes have been summarized in \cite{quantumAdvantage}. If both the clique-density requirement and the inverse-polynomial eigenvalue threshold are fulfilled, the algorithm can estimate the normalized Betti number up to inverse-polynomial additive error $\varepsilon$ in polynomial time. In this work, we give complexity-theoretic and analytical evidence that even in this optimal regime, the additive accuracy achievable by quantum algorithms for TDA is not sufficient to solve the problem \textbf{Homology}, e.g., to decide whether the Betti numbers are zero or non-zero. 

What practical information does an additive approximation to Betti numbers contain? We first note practical applications of TDA \cite{zomorodian2012topological,wasserman2018topological} require the actual Betti numbers $\beta_k$. More importantly, the requirement that the simplicial complex is clique-dense precisely means that the denominator of $\zeta_k$, $|S_k|$, scales exponentially with $n$. Unless the Betti number $\beta_k$ scales similarly, the normalized Betti number $c_k$ will be exponentially small. In that case, any inverse-polynomial additive error $\varepsilon$ cannot distinguish it from zero.

We will show in section \ref{sec:bettirandom} that Betti numbers typically do not scale exponentially in the number of vertices. Indeed, they are asymptotically almost always bounded by $\mathcal{O}(n)$, independent of $k$. The LGZ-algorithm thus requires exponential time to distinguish the Betti numbers from zero, which could be seen as a consequence of Theorem \ref{thm:Homology}. 

Because of the above discussion, we argue that it can be misleading to state the runtime in terms of the normalized \footnote{The computational complexity of computing normalized Betti numbers is an interesting topic by itself. It was investigated in \cite{quantumAdvantage}, which we briefly review in section 6.} Betti numbers $c_k$. Instead, a more insightful description of the algorithm is the runtime required to compute the actual Betti number $\beta_k$ up to multiplicative error $\delta$. This can be achieved naively by choosing $\varepsilon = \delta \frac{\beta_k}{|S_k|}$, resulting in \begin{equation}
\mathcal{O}\left(\frac{|S_k|^2}{\beta_k^2}\frac{(n^3\kappa  +nk^2\zeta_k^{-1/2})}{\delta^{2}}\right).\label{rt:naive}
\end{equation} A series of recent improvements to the original LGZ-algorithm, most notably by utilizing quantum counting \cite{LloydAlgo,reviewBetti} and an efficient construction of the boundary operator \cite{linearDepth} via fermionic annihilation operators, yield an algorithm that naturally outputs a multiplicative approximation with a substantially improved runtime of\begin{equation}
    \mathcal{O}\left( \frac{1}{\delta} \left(n^2\sqrt{\frac{{n\choose k+1}}{ \beta_k}} + n\kappa\sqrt{\frac{|S_k|}{\beta_k}}\right)\right).
    \label{rt:actual}
\end{equation} We note that further improvements to the complexity \eqref{rt:actual} might be possible by e.g. combining the results of \cite{linearDepth,logDepth,persistent,quantumAdvantage,reviewBetti,GoogleBettiBerry,AmazonBettiMcArdle}. However, it is to the best of our knowledge not possible to prepare the state $\rho_k$ significantly faster than $\Omega(\zeta_k^{-1/2})$ or estimate the Betti number $\beta_k$ given $\rho_k$ significantly faster than $\Omega\left(\sqrt{\frac{|S_k|}{\beta_k}}\right)$, hence the total runtime is lower bounded by \begin{equation}
    T_q = \Omega\left(\sqrt{\frac{{n\choose k+1}}{\beta_k}}\right).
\end{equation} While this might suggest a polynomial runtime for large Betti numbers, we will show in section \ref{sec:bettirandom} that this expression is exponential in $n$ for asymptotically almost all inputs, in line with Theorem \ref{thm:Betti} and Theorem \ref{thm:Homology}. We now prove both theorems.
\section{Computing Betti Numbers is \#P-hard}
\label{sec:counthard}
In this section, we prove the complexity-theoretic hardness of computing Betti numbers exactly. Let us recall the problem definition and theorem from the first section.

\begin{problem*}[\textbf{Betti}]
Given a clique complex $S$ defined by its vertices and edges and an integer $k\geq0$ as input, output the $k$-th Betti number $\beta_k$ of $S$.
\end{problem*}
\begingroup
\def\thetheorem{\ref{thm:Betti}}
\begin{theorem}
The problem \emph{\textbf{Betti}} is \emph{(a)} \#P-hard and \emph{(b)} remains \#P-hard when the input is restricted to clique-dense complexes.
\end{theorem}
\addtocounter{theorem}{-1}
\endgroup
\noindent Informally, the complexity class NP is the set of decision problems for which a solution can be verified in polynomial time. The complexity class \#P is the set of counting problems associated to NP. That is, whereas a problem in NP asks whether a given instance has a solution or not, the corresponding problem in \#P asks how many solutions the instance has. The archetypal \#P-complete problem is \#SAT, the problem of counting the number of satisfying truth assignments to a given SAT instance (more details on SAT can be found after Theorem 3). Clearly, solving a problem in \#P is at least as hard as solving the corresponding problem in NP, as computing the number of solutions in particular determines whether this number is non-zero. On the other hand, there are problems where the decision version is easy, i.e. in P, but the counting version is \#P-complete. The most prominent such example is the matrix permanent \cite{valiant1979}. No known quantum algorithm solves an NP-complete problem in polynomial time and it is conjectured that NP $\not \subset $ BQP, i.e. that NP-hard problems are not accessible to quantum computers. The same holds for \#P-hard problems. It is worth noting that it is not known whether the NP-hardness of a problem implies the \#P-hardness of the corresponding counting problem \cite{livne2009note}.

The main idea of our proof of Theorem \ref{thm:Betti} is to relate the computation of Betti numbers to the computation of the number of maximal cliques, which is a known \#P-hard problem. To do so, we introduce the \emph{Euler characteristic} $\chi(S)$ of a simplicial complex $S$ defined on $n$ vertices as \begin{equation}
    \chi(S) = \sum_{k = 0}^{n-1} (-1)^k |S_k|,
\end{equation} where as before $|S_k|$ is the number of $k$-simplices in $S$. The index $k$ is bounded by $n-1$ because the largest possible simplices on $n$ vertices are $(n-1)$-simplices. The Euler characteristic $\chi$ allows us to 
relate the task of computing $\beta_k$ to the task of computing $|S_k|$ via the Euler-Poincaré formula \cite{EulerPoinc}:
\begin{theorem}[Euler-Poincaré]
The Euler characteristic of a simplicial complex $S$ has an alternative representation as $$\chi(S) = \sum_{k = 0}^{n-1} (-1)^k \beta_k,$$ where $\beta_k$ are the Betti numbers of $S$.
\end{theorem}
\begin{proof}
By the Rank-nullity theorem, we have $|S_k| = \dim \text{Im } \del_k + \dim \text{Ker } \del_k$ and by the definition of the Betti number $\beta_k = \dim \hal_k =\dim \text{Ker } \del_k - \dim \text{Im } \del_{k+1}.$ Thus
\begin{align*}
    &\sum_{k= 0}^{n-1} (-1)^k(|S_k|-\beta_k)\\ &=\sum_{k = 0}^{n-1} (-1)^k (\dim \text{Im } \del_k + \dim \text{Im } \del_{k+1}) = 0,
\end{align*} where we have used $\text{Im } \del_0 =  \text{Im } \del_{n} = 0$. 
\end{proof}
The computation of the Euler characteristic thus reduces to the computation of the Betti number vector. In the following, we show that $\#$SAT, the problem of counting the number of solutions to a given SAT problem, reduces to the computation of the Euler characteristic of a Vietoris-Rips complex defined by its vertices and edges. Recall that a $SAT$ (or \emph{boolean satisfiability problem}) instance with variables $X_1,\dots,X_n$ and clauses $C_1,\dots,C_k$ in conjunctive normal form (CNF) is a boolean formula $C_1 \land C_2 \land \dots \land C_k$ where every clause $C_j$ is a disjunction of literals. A literal is either a variable $X_i$ or its negation $\neg X_i$.
The hardness of computing the Euler characteristic of an abstract simplicial complex defined by its inclusion-maximal simplices has already been established in \cite{complexEchar}. However, this proof does not directly apply here, because the complexes considered in \cite{complexEchar} are not clique complexes and a much more extensive description of the complex is given as input. Indeed, translating between both descriptions of a clique complex, i.e. computing the list of inclusion-maximal simplices from a list of edges and vertices, is by itself a \#P-hard problem \cite{Karp1972}. Our techniques are inspired by a proof of the hardness of computing independent sets \cite{complexEchar}, \cite{strongNP}. 
\begin{proof}[Proof of Theorem \ref{thm:Betti}(a)]
Let $\mathcal{K}$ be a $SAT$ instance in conjunctive normal form with variables $X_1, \dots,X_n$ and clauses $C_1, \dots , C_k$. We will construct a clique complex $S$ on $3n+k$ vertices such that the number of satisfying truth assignments of $\mathcal{K}$ equals $(-1)^{n}(1- \chi(S)).$

\begin{figure}
    \centering
    \includegraphics[width=0.7\linewidth]{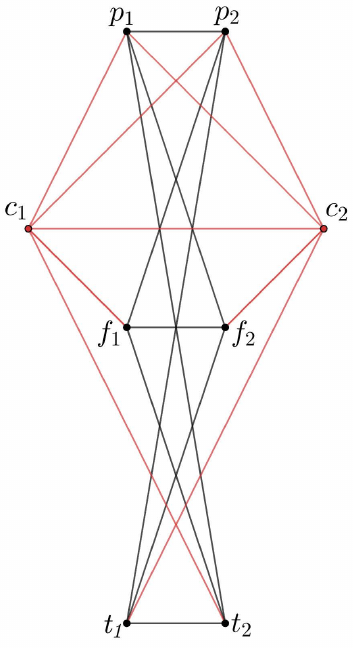}
    \caption{The graph $G$ corresponding to the $SAT$-instance $\mathcal{K} = (X_1 \lor \neg X_2) \land (\neg X_1 \lor X_2).$ The edges connected to the two clauses $c_1$ and $c_2$ are colored in red for clarity. The Vietoris-Rips complex on $G$ has 8 vertices, 18 edges, 10 triangles (2-simplices) and 1 tetrahedron (3-simplex). Its Euler characteristic is thus $\chi(S)=8-18+10-1=-1$ and the number of satisfying truth assignments $(-1)^2(1-\chi(S)) = 2$. }
    \label{fig:graph1}
\end{figure}
To construct $S$ we first build up the underlying graph $G$ as follows: For every variable $X_i$, construct 3 vertices $t_i, f_i$ and $p_i$. It will become clear later that $t_i$ and $f_i$ encode whether $X_i$ is \emph{true} or \emph{false}, and $p_i$ whether $X_i$ appears\footnote{If every literal $X_j$ appears in the $SAT$-instance, the vertices $p_j$ do not contribute to the Euler characteristic and they can be neglected, leading to a Graph of size $2n+k$.}  in $\mathcal{K}$. Connect every two points $v_i, v_j$ for which $i\not=j$. Here, $v \in \{t,f,p\}$ is a placeholder. Consequently, add an additional vertex $c_j$ for every clause $C_j$. Connect $c_j$ to every vertex \textbf{\emph{except}} to \begin{itemize}
\item $t_i$ if the clause $C_j$ contains the literal $X_i$.
\item $f_i$ if the clause $C_j$ contains the literal $\neg X_i$.
\end{itemize} In Figure \ref{fig:graph1} we illustrate the graph $G$ corresponding to the 2-$SAT$ instance $\mathcal{K} = (X_1 \lor \neg X_2) \land (\neg X_1 \lor X_2).$ The vertex set $V$ of $G$ is \begin{equation}
    V = \{t_i,f_i,p_i \ |\ 1 \leq i \leq n \} \cup \{c_j \ |\ 1 \leq j \leq k\}.
\end{equation}We can now build a clique complex $S=Cl(G)$ on $G$ in the usual way by identifying ($k+1$)-cliques with $k$-simplices. It remains to show that the Euler characteristic of $Cl(G)$ counts the number of solutions to $\mathcal{K}$. To do so, we introduce the notion of a \emph{maximal clique,} which is a clique $s \subset S$ that is not contained in any other clique of $S$. Maximal cliques should not be confused with maximum cliques, which are cliques of maximal cardinality. While a maximum clique is always maximal, the converse must not hold.
For the clique complex shown in Figure \ref{fig:graph1}, the maximal cliques are \begin{align*}
    &\{c_1,p_1,p_2,c_2\} \\ &\{c_1,p_1,t_2\},\{c_1,p_2,f_1\},\{c_1,f_1,t_2\},\\ &\{c_2,p_2,t_1\},\{c_2,p_1,f_2\}, \{c_2,f_2,f_1\} \\ &\{t_1,t_2\},\{f_1,f_2\}.
\end{align*}
Let $A$ be the set of all vertices labeled either $t_i$ or $f_i$ for any $i$ and let $B$ be the set of all vertices labeled either $c_i$ or $p_i$. By construction, $A\cup B = V$. Let further $\Gamma(S)$ be the set of all maximal cliques of $S$ that contain only vertices of $A$. In the above example, $\Gamma(S) = \{ \{t_1,t_2\},\{f_1,f_2\}\}$ We will show the following two results: \begin{itemize}
    \item[\textbf{i)}] The elements of $\Gamma(S)$ are in a one-to-one correspondence with the solutions of the $SAT$ instance $\mathcal{K}$.
    \item[\textbf{ii)}] All remaining cliques $s \in S\backslash\Gamma(S)$ contribute trivially to the Euler characteristic, that is $$ \chi(S\backslash \Gamma(S)) := \sum_{s \in S\backslash \Gamma(S)} (-1)^{|s|-1} = 1.$$
\end{itemize}
As a consequence of \textbf{i)} and \textbf{ii)}, the Euler characteristic of $S$ equals \begin{align}
  \label{interEuler}
    \chi(S) &= \sum_{s\in S}(-1)^{|s|-1} \\ &= (-1)^{n-1}(\textbf{number of solutions to }\mathcal{K}) + 1 \nonumber, 
\end{align} which proves Theorem \ref{thm:Betti}. It remains to show \textbf{i)} and \textbf{ii)}.

\noindent\textbf{Proof of i).} Let $s \in \Gamma(S)$. Suppose there exists a fixed $i$ such that $s$ contains neither $t_i$ nor $f_i$. Then $s$ can be extended to the strictly larger clique $s \cup \{p_i\}$, since $p_i$ is connected to all other $t_j,f_j$ for $j\not = i$. This contradicts that $s$ is maximal.
Likewise, $s$ can not contain both $t_i$ and $f_i$ for any fixed $i$, as they are not connected in $S$. Therefore, any $s\in \Gamma(S)$ has cardinality $n$ and corresponds to an assignment of the variables $X_1,\dots, X_n$. Consider now a clause $c_j$. By the maximality of $s$, there is at least one vertex of $s$ that is not connected to $c_j$ and therefore the truth assignment corresponding to $s$ fulfills the clause $c_j$. Since $c_j$ was arbitrary, $s$ satisfies $\mathcal{K}$. The same reasoning shows that any $z\in S$ corresponding to a satisfying truth assignment of $\mathcal{K}$ is maximal and contained in $\Gamma$.

\noindent\textbf{Proof of ii).} Let $\widetilde{L}\subset S$ be the set of all cliques of $S$ that only contain vertices of $A$ and consider $L = \widetilde{L}\backslash \Gamma(S)$. For each $v\in L$, let us define the set \begin{equation}
    D_v = \{k \in S \ | \ v \subset k \text{ and } k\backslash v \subset B\}
\end{equation} of all extensions of $v$ to $B$. Note that $\{D_v\}_{v\in L}$ together with $D_{\emptyset}=\{k\in S | k \subset B\}$ is a complete partition of $S\backslash\Gamma(S)$, that is, every $s \in S\backslash\Gamma(S)$ is contained in exactly one $D_v$. Let $v$ be a clique contained in $L$. Thus $v$ is not maximal and by \textbf{i)} we can find a $p_j \in B$ such that $v \cup \{p_j\}$ is in $D_v$. Since all elements of $B$ are pairwise connected, there is a bijection between the sets $E_v = \{k\in D_v \ | \ p_j \not \in k \}$ and $F_v = \{k\in D_v \ | \ p_j \in k \}$: \begin{equation}
    E_v \to F_v, \quad k \mapsto k\cup \{p_j\}.
\end{equation}
Because $E_v \cup F_v = D_v$ by definition, the Euler characteristic of $D_v$ is $\chi(D_v) = \chi(E_v)+\chi(F_v) = 0$. Using once again that all elements of $B$ are pairwise connected and the fact that the clique complex of a complete graph has Euler characteristic $1$, we get $\chi(D_\emptyset) = 1$ and 
\begin{equation}
    \sum_{s \in S\backslash \Gamma(S)} (-1)^{|s|-1} = \chi(D_\emptyset) + \sum_{v\in L} \chi(D_v)  = 1,
\end{equation}which proves \textbf{ii)}. Because the elements of $\Gamma(S)$ have cardinality $n$, they contribute with a factor $(-1)^{n-1}$ to $\chi(S)$. This leads to expression \eqref{interEuler} and concludes the proof.
\end{proof}
Before proving Theorem \ref{thm:Betti}(b), we mention two alternative approaches to arrive at the \#P-hardness of computing Betti numbers exactly. First, note that the only cliques contributing to the Euler characteristic of the above graph $G$ are the cliques of size $n$. Building upon this idea, one can show that computing the $n$-th Betti number of $Cl(G)$ is already \#P-hard, as opposed to computing all Betti numbers. Because we do not require this stronger statement, we omit the details. Second, one could also arrive at Theorem \ref{thm:Betti}(a) and \ref{thm:Betti}(b) building upon a result  \cite{crichigno2020supersymmetry} that computing the Witten index of a certain local supersymmetric system is \#P-hard. This is because the Witten index can be related to the independence polynomial (which is directly related to the Euler characteristic) of a clique complex via the so-called Fermionic hard-core model. 

We can strengthen Theorem \ref{thm:Betti}(a) by restricting the class of simplicial complexes allowed as input to the problem \textbf{Betti}. As argued in section 2, the LGZ quantum algorithm works best in the regime of clique-dense complexes. We now show that computing Betti numbers remains hard even in this regime. 
\begin{proof}[Proof of Theorem \ref{thm:Betti}(b)]
In the proof of Theorem \ref{thm:Betti}(a) we reduced an arbitrary \#$SAT$ instance to the computation of the Euler characteristic of the clique complex of a corresponding graph $G$. We will show that $Cl(G)$ is already clique-dense, which evidently proves Theorem \ref{thm:Betti}(b). Recall that for an instance $\mathcal{K}$ with $n$ variables and $s$ clauses, $G$ is defined on $N = 3n+s$ vertices. If $\mathcal{K}$ belongs to $\ell$-$SAT$, i.e. every clause $C \in \mathcal{K}$ has $\ell$ literals, the number of edges of $G$ is \begin{equation}
        |E(G)| = \frac{9}{2} n (n-1) + 3ns - ls + \frac{1}{2}s(s-1). 
    \end{equation}
Because $SAT$ reduces to $3$-$SAT$, we can fix $\ell = 3$. Moreover, as $3$-$SAT$ remains NP-hard \footnote{Technically, we not only require that the problem remains NP-hard but that the corresponding counting problem remains \#P-hard. However, these reductions are parsimonious and preserve the number of solutions. Thus they also establish \#P-hardness of the counting problem.} when restricted to instances where every variable appears 4 times \cite{tovey1984simplified,berman2007computational}, we can fix $s = \frac{4}{3}n$. The fraction $\gamma$ of connected edges over the number of vertices squared is then \begin{align}
    \gamma &= \frac{|E(G)|}{N^2} = \frac{1}{2}-\frac{9n+7s}{2N^2} \nonumber\\&= \frac{1}{2}\left(1-\frac{165}{169n}\right) > \frac{1}{2}\left(1-\frac{1}{n}\right).
\end{align} The clique-density theorem \cite{cliqueDense} states that the number of $k$-simplices in $G$ is lower bounded by $\Omega(N^{k+1})$ if \begin{equation}
    \gamma > \frac{1}{2}\left(1-\frac{1}{k}\right).
\end{equation}
Thus the clique-density condition \eqref{cliqueDense} is fulfilled for any $k \leq n$, in particular for $k = \Theta(n)$, and $Cl(G)$ is clique dense.
\end{proof}
This concludes the proof of Theorem \ref{thm:Betti}. We emphasize that it is not clear whether either  \textbf{Betti} or \textbf{Homology} are in \#P or NP, respectively, which would make them \#P- or NP-complete. The non-triviality of the $k$-th Betti number could be witnessed by a cycle that is not a boundary, but such a cycle may be exponentially large in the size of the input. 
\section{Approximating Betti Numbers is NP-hard}
\label{sec:approxhard}
In the previous section, we established that computing Betti numbers exactly is \#P-hard. However, the LGZ-algorithm only approximates Betti numbers. To be able to argue about the runtime of quantum algorithms for TDA, we need to determine the complexity of computing a multiplicative approximation to $\beta_k$. This motivates us to define the following problem, which we recall along with Theorem \ref{thm:Homology} from section 1.
\begin{problem*}[\textbf{Homology}]
Given a clique complex $S$ defined by its vertices and edges and an integer $k\geq0$ as input, output \emph{true} if $\beta_k > 0$ and \emph{false} if $\beta_k = 0.$
\end{problem*}
\begingroup
\def\thetheorem{\ref{thm:Homology}}
\begin{theorem}
The problem \emph{\textbf{Homology}} is \emph{(a)} NP-hard and \emph{(b)} remains NP-hard when the input is restricted to clique-dense complexes.
\end{theorem}
\addtocounter{theorem}{-1}
\endgroup
\noindent Theorem \ref{thm:Homology}(a) is due to Adamaszek and Stacho \cite{adamaszek2016}, who proved that the decision problem \textbf{Homology} is NP-hard for the clique complexes of co-chordal graphs. Since the class of all graphs in particular includes the class of co-chordal graphs, the theorem follows. This by itself however is not enough to reason about the runtime of the LGZ-algorithm, as co-chordal graphs tend to be sparse, whereas it has already been established that a necessary requirement for a polynomial runtime of the LGZ-algorithm is that the input is clique-dense. To strengthen the results of Adamaszek and Stacho, we show that the problem \textbf{Homology} retains its hardness when restricted to the class of clique-dense complexes.
\begin{proof}[Proof of Theorem \ref{thm:Homology}(b)]
Let $G = (V,E)$ be any co-chordal graph with vertex set $V=\{x_1,\dots,x_n\}$ and edges $E$. Building upon constructions in \cite{adamaszek2016,StarCluster}, we will relate the homology of $Cl(G)$ to the homology of the complement $\bar{H}$ of a bipartite graph $H$ defined via the Alexander-dual of $G$. We will then show that $Cl(\bar{H})$ is clique-dense, which reduces the homology problem of co-chordal graphs to the homology problem of clique-dense complexes and proves the theorem.

Denote by $Cl(G)$ the clique complex of the co-chordal graph $G$. Consider the complex $A_G$ defined on the same vertex set $V$ consisting of simplices $s$ with \begin{equation}
s\in A_G \iff (V\backslash s) \not \in Cl(G).
\end{equation} $A_G$ is called the Alexander-dual of $G$ and consists of the complements of all non-cliques in $Cl(G).$ The homologies of $A_G$ and $Cl(G)$ are related via the Alexander duality \cite{AlexanderDual} \begin{equation}
   H_{k}(Cl(G))  = H^k(Cl(G)) = H_{n-k-3}(A_G),
   \label{firsthom}
\end{equation}where the first equality holds due to Corollary 2.2 in \cite{adamaszek2016}. $A_G$ is not a clique complex and thus not the type of complex that appears in TDA. However, we can relate the homology of $A_G$ to the homology of the clique complex of a bipartite graph $H$, as follows. Denote by $\mathcal{M}(A_G)$ the set of inclusion-maximal faces of $A_G$. By construction, \begin{equation}
    \mathcal{M}(A_G) = \{V \backslash e \ |\  e \in V\times V \text{ and } e \not \in E\}
    \label{maximalFaces}
\end{equation}is the complement of all edges $e$ that are not in the edge set of $G$, i.e. that are edges of $\bar{G}$. In particular, $\mathcal{M}(A_G)$ can be computed in polynomial time given $G$. Let us now construct a bipartite graph $H$ with two parts of size $n$ and $m$ respectively, where $m= |\mathcal{M}(A_G)|$. Each vertex in the first part corresponds to one of the vertices $x_i$ of $G$, and each vertex in the second part corresponds to the inclusion-maximal faces $F_j \in \mathcal{M}(A_G)$ of the Alexander dual of $G$. Connect every pair $(x_i,F_j)$ for which $x_i \not \in F_j$. By equation \eqref{maximalFaces}, each $F_j$ is only connected to exactly two vertices. $H$ is thus a graph on $n+m$ vertices with $2m$ edges and clearly sparse. By construction, $Cl(\bar{H})$ is isomorphic to the suspension of $A_G$ and the suspension isomorphism \cite{adamaszek2016,StarCluster,kozlov2008combinatorial} relates the homology of both complexes via \begin{equation}
    H_k(A_G) = H_{k+1}(Cl(\bar{H})),
    \label{sechom}
\end{equation}where the complement $\bar{H}$ is now dense. Combining equation \eqref{firsthom} and \eqref{sechom} yields \begin{equation}
    H_k(Cl(G)) = H_{n-k-2}(Cl(\bar{H})),
    \label{homequiv}
\end{equation} which reduces the homology problem of co-chordal graphs $G$ to the homology problem of dense co-bipartite graphs $\bar{H}$. For equation \eqref{homequiv} to be well-defined, we have to pick $0 \leq k\leq n-2$, which does not matter for the reduction as the last Betti number $\beta_{n-1}$ can be computed in linear time.

To conclude the proof, we show that $Cl(\bar{H})$ obtained via this construction is clique-dense.
Let $N = n+m$ be the number of vertices of $\bar{H}$.  The fraction $\gamma$ of connected
edges over the number of vertices squared is \begin{align}
    \gamma &= \frac{|E(\bar{H})|}{N^2} = \frac{N(N-1) - 4m}{2N^2} \nonumber\\&= \frac{1}{2}\left(1-\frac{5N-4n}{N^2}\right) > \frac{1}{2}\left(1-\frac{5}{N}\right).
\end{align}
By the clique-density theorem \cite{cliqueDense}, $Cl(\bar H)$ contains $\Omega(N^{k+1})$ many $k$-simplices if \begin{equation}
\gamma > \frac{1}{2}\left(1-\frac{1}{k}\right).
\label{cliqueEq}
\end{equation} Hence $\bar{H}$ fulfills the requirement \eqref{cliqueDense} for any $k\leq N/5$ and in particular, since $n\leq N$, for $k = \Theta(n).$ The clique complex $Cl(\bar{H})$ is thus clique-dense.
For completeness, we note that one could similarly show that \eqref{cliqueDense} is fulfilled for any $k\leq n$ by utilizing that the MaxClique problem required to establish NP-hardness for co-chordal graphs \cite{adamaszek2016} remains hard for graphs with a constant fraction of edges, in which case $m = \Theta(n^2)$ and \begin{equation}
    \frac{5N-4n}{N^2} = o\left(\frac{1}{n}\right).
\end{equation} Having reduced the homology problem of the clique complexes of co-chordal graphs to the homology problem of clique-dense clique complexes, this concludes the proof.
\end{proof}

Theorems \ref{thm:Betti} and \ref{thm:Homology} together describe the computational hardness of topological data analysis. It is widely believed that quantum computers cannot solve NP-hard or \#P-hard problems in subexponential time. Our results therefore imply that any quantum algorithm for TDA, including the LGZ-algorithm, runs in exponential time in the worst case. Moreover, our proofs show that even in the optimal clique-dense regime, quantum algorithms for TDA require exponential time. Because the LGZ-algorithm performs optimally in the clique-dense regime, and the relative size of this regime tends towards zero, this further implies that the LGZ-algorithm runs in exponential time for asymptotically almost all inputs. The same holds for any quantum algorithm for TDA that uses similar techniques (Grover search) as LGZ to arrive at the mixed state $\rho_k$, which are only efficient in the clique-dense case. 


\section{Time Complexity of the LGZ-Algorithm}
\label{sec:bettirandom}
Having established the complexity-theoretic hardness of topological data analysis, we turn back to the original LGZ-algorithm. We consider here the version of the LGZ-algorithm that outputs a multiplicative approximation to $\beta_k$, as originally intended \cite{LloydAlgo} and required by practical applications. Our previous results imply that the algorithm will have a runtime exponential in the input size for asymptotically almost all inputs. In this section, we verify this directly by considering random Vietoris-Rips complexes.

Disregarding other details - like potentially exponentially small eigenvalue gaps of the simplicial Laplacian - the optimal runtime of the LGZ-algorithm as derived in equation \eqref{rt:actual} will be of order $\text{poly}(n) \xi$, where \begin{equation}
    \xi = \sqrt{\frac{{n\choose k+1}}{ {\beta_k}}}.
\end{equation} Recall that the best known classical algorithm for computing $\beta_k$ takes time $\mathcal{O}\left({n\choose k+1}\right)$. The LGZ-algorithm will thus achieve a super-quadratic speedup only if the $k$-th Betti number is large. However, as we will discuss in the following, Betti numbers of Vietoris-Rips complexes are small in expectation and variance.

Consider a set of points $\{x_i\}_{i=1}^n$ randomly distributed i.i.d. in a metric space of dimension $d$ according to any bounded probability density function. Let us denote by $V^\epsilon$ the simplicial complex that results from the TDA procedure outlined in section 2. That is, we choose a grouping scale $\epsilon > 0$, connect all $(x_i,x_j)$ that are closer than $\epsilon$ and identify cliques in the corresponding graph with simplices. The resulting simplicial complex is called a Vietoris-Rips complex, and together with the assumption that $\{x_i\}_{i=1}^n$ were distributed randomly, a \emph{random Vietoris-Rips} complex. Note that this is the most natural input distribution model for practical applications of topological data analaysis.

The Betti numbers of random Vietoris-Rips complexes were analyzed by Kahle in \cite{Kahle_2011,LimitBetti}. Because we are interested in the behaviour of Betti numbers as $n \to \infty$, we allow the grouping scale $\epsilon$ to vary with $n$. Let us denote by $\hat{r} = n^{-1/d}$ the average distance \footnote{By possibly normalizing the metric, we assume that no two points of $V$ are further than 1 apart. Hence $V^\epsilon$ is disconnected for $\epsilon = 0$ and fully connected for $\epsilon = 1$.} between two vertices of $V^\epsilon$. Kahle identified three regimes for $\epsilon$ that fully describe the asymptotic behaviour of Betti numbers of $V^\epsilon$. 
\begin{itemize}
    \item \textbf{Subcritical regime:} The subcritical regime occurs when $\epsilon = o(\hat{r})$ and the Vietoris-Rips complex is sparsely connected. In this regime, the Betti numbers experience a transition from vanishing to non-vanishing. In either case, they are upper bounded by \begin{equation}
        \mathbb{E}[\beta_k] = o(n).
    \end{equation}
    \item \textbf{Critical regime:} If $\epsilon = \Theta(\hat r)$ is in the same order of magnitude as $\hat r$, $V$ is said to be in the critical regime. This region is also called the thermodynamic limit and is the regime where Betti numbers take on their maximal value. The expectation value and variance of $\beta_k$ are\begin{equation}
    \mathbb{E}[\beta_k] = \Theta(n), \quad \text{Var}[\beta_k] = \Theta(n).
\end{equation}
\item \textbf{Supercritical regime:} Here, $\epsilon = \omega(\hat r)$ dominates the average separation $\hat r$ and large connected components emerge which do not contribute to the homology of V. Consequently, Betti numbers grow sublinearly and \begin{equation}
    \mathbb{E}[\beta_k] = o(n).
\end{equation}
\end{itemize}
In all three regimes, the expectation value of $\beta_k$ grows at most linearly with $n$. Moreover, the variance Var$(\beta_k)$ is also bounded by $\mathcal{O}(n)$. Thus $\beta_k = \mathcal{O}(n)$ asymptotically almost surely (a.a.s.) over random Vietoris-Rips complexes and the LGZ-algorithm applied to random Vietoris-Rips complexes almost always achieves only a quadratic speedup over classical approaches.

This average-case analysis does of course not exclude the possibility that specific families of graphs exhibit large enough Betti numbers to provide a much better quantum advantage. Indeed, a recent paper by Google Quantum AI \cite{GoogleBettiBerry} provided an example of a graph with a superpolynomial speedup. 
It has Betti numbers scaling as $\beta_{\log(n)-1} \sim (n/\log(n))^{\log(n)},$ which is sufficiently large to give a superpolynomial advantage over the best known classical algorithm. Unlike the present paper, the authors of \cite{GoogleBettiBerry} study the regime where $k$ scales logarithmic with $n$. This relaxes the clique-density condition \eqref{cliqueDense} but removes the possibility for an exponential speedup. To the best of our knowledge, there is currently no known example of a graph that provides an exponential quantum advantage.

Upper bounds on the size of Betti numbers of Vietoris-Rips complexes imply lower bounds on the runtime of the LGZ-algorithm. Goff analyses the maximal size of Betti numbers of Vietoris Rips complexes in \cite{ExtremalBetti}. For the $k$-th Betti number, he proves an upper bound of $\mathcal{O}(n^k)$, which is slightly more prohibitive than the trivial upper bound $\mathcal{O}(n^{k+1}).$ However, the best construction in \cite{ExtremalBetti} only achieves $\beta_k = \mathcal{O}(n^{k/2+1/2})$ (which corresponds to a quartic quantum advantage) and does not saturate this bound. Since the upper bound in \cite{ExtremalBetti} was obtained inductively, it might be loose. Whether or not it can be tightened is an open question that would provide insight into whether any super-polynomial quantum advantage for TDA can be achieved in the regime where $k$ scales as $k = \Theta(n)$.

The above discussion implies that the LGZ-algorithm achieves a Grover-like speedup over the best classical approach in almost all cases and up to a superpolynomial speedup in certain specific cases. This agrees with Theorem \ref{thm:Homology} and the widely-believed conjecture that quantum computers cannot solve \#P-hard or NP-hard problems in polynomial time. While a quadratic speedup is significant, this suggests that quantum algorithms for topological data analysis might not achieve a practical advantage for NISQ-era and early generation fault-tolerant quantum computers \cite{babbush2021focus}. In the next section, we argue that an exponential speedup can be recovered if computational problems beyond TDA are considered.
\section{Quantum Advantage Beyond TDA}
Up to this point, our complexity-theoretic analysis in section 3 and 4 as well as our direct analysis of the LGZ-algorithm in section 5 have exclusively focused on clique complexes, and even more specifically on Vietoris-Rips complexes. While all simplicial complexes appearing in the typical setting of TDA\footnote{In theory, a set of data points could also be represented by the related Čech complex. This representation has the advantage that it is homotopy equivalent to the union of $\epsilon$-balls centered around each data point. However, the construction of Čech complexes is computationally much more expensive than the construction of Vietoris-Rips complexes, hence they are not often used in practice. They also have similarly scaling Betti numbers as Vietoris-Rips complexes \cite{Kahle_2011}, thus we expect similar prohibitive bounds on quantum speedups.} are of this form, the LGZ-algorithm and the ideas behind it can be applied to a much broader class of simplicical complexes. Computing Betti numbers of more general simplicial complexes is an important problem in computational topology \cite{edelsbrunner2022computational,basu2005betti,bayer1998extremal,romer2001generalized}. In this section, we argue that an exponential quantum advantage can be recovered if certain computational problems beyond TDA are considered. To illustrate this, we emphasize three points:
\begin{itemize}
\item The limiting factor for quantum advantage in TDA is the construction of (and subsequent random sampling from) the list of $k$-simplices of a clique complex, given the underlying graph as input. This by itself is an NP-hard problem, thus quantum algorithms solve it only polynomially faster than the best classical algorithm.
\item Given an oracle that provides random sampling from the $k$-simplices in an (abstract) simplicial complex, quantum algorithms for TDA can efficiently compute the $k$-th Betti number of the simplicial complex. Moreover, they can do so sometimes exponentially faster than the best known classical algorithm with access to the same oracle. 
\item Such an oracle is not natural in TDA. Indeed, constructing such an oracle for a clique complex described by its underlying graph is again NP-hard. However, in some applications of computational topology, it is natural to describe simplicial complexes by a method of specifying the set of simplices. This enables the efficient construction of such an oracle.
\end{itemize}

Let us illustrate the first point. The main complexity of quantum algorithms for TDA arises from two factors. First, the clique complex $S$ has to be constructed from the input graph. We call this the \emph{construction} step. Taking again $k = \Theta(n)$ and surpressing factors polynomial in $n$, this step takes time \begin{equation}
   \mathcal{O}(\zeta_k^{-1/2}) = \mathcal{O}\left( \sqrt{\frac{{n\choose k+1}}{|S_k|}}\right)
    \label{fr:complex}
\end{equation} via Grover search, where $|S_k|$ is the number of $k$-simplices in the complex and $n$ is the number of vertices. The complex construction by itself is a limiting factor in the runtime: Efficient sampling from the $k$-simplices of a clique complex, given a description of the underlying graph, allows us to find a maximum clique of the graph. It is therefore an NP-hard problem, which in particular suggests that quantum computers cannot solve it significantly more efficiently than via Grover search. 

The second step, which we call \emph{estimation} step, is to estimate the Betti numbers given the ability to sample from the set of $k$-simplices, and to determine whether a given $k$-simplex is in the complex.   The estimation step can be done via, e.g., quantum phase estimation and quantum counting, which has complexity proportional to \begin{equation}
     \mathcal{O}\left(\frac{1}{\delta}\sqrt{\frac{|S_k|}{\beta_k}}\right).
        \label{fr:betti}
\end{equation} Here, $\delta$ is the multiplicative error and factors of poly($n)$ are again suppressed. It is this second step where quantum algorithms for TDA have the potential to provide an exponential advantage over classical algorithms. To our knowledge, the best classical algorithm can compute the Betti numbers of a simplicial complex with inverse-polynomial spectral gap of the Laplacian $\frac{\lambda_{max}}{\lambda_{min}} = \text{poly}(n)$ from a list of simplices in time $\mathcal{O}(|S_k|)$. This is because determining the kernel dimension of sparse matrices scales at best linearly in the matrix dimension \cite{ubaru2016fast}. For a detailed discussion of the best classical algorithms for estimating Betti numbers, see for example Section IV in \cite{GoogleBettiBerry}.

On the other hand, the LGZ-algorithm can solve the same task up to constant multiplicative error exponentially faster if $\beta_k \approx |S_k|$. In the following subsection, we discuss regimes where the quantum algorithm solves the second step, \emph{estimation}, efficiently. 

\subsection{Complexes beyond Vietoris-Rips}
While Vietoris-Rips complexes do not have sufficiently large Betti numbers for a large quantum speedup, other types of complexes do. For Erdös-Renyi complexes, the LGZ-algorithm solves the \emph{estimation} step exponentially faster than the best known classical approach in asymptotically almost all cases. Random Erdös-Renyi complexes are the clique complexes of random Erdös-Renyi graphs $G(n,p)$. The graph $G(n,p)$ on $n$ vertices with probability parameter $p$ is constructed by randomly connecting vertices with probability $p$, independently from every other edge. Kahle \cite{SizeofRandom} showed that for random Erdös-Renyi complexes, the fraction in \eqref{fr:betti} is asymptotically almost always 1. Specifically, for any $0\leq k<n$ and a probability parameter $n^{-1/k} < p < n^{-1/(k+1)},$ we have \begin{equation}
    \mathbb{E}(\beta_k) \approx \mathbb{E}(|S_k|) \quad \text{ and } \quad \beta_k \approx |S_k|,
\end{equation} i.e. the $k$-th Betti number takes on its maximal possible value $|S_k|$ (the number of $k$-simplices in the complex). Consequently, the estimation step takes time $\mathcal{O}(\text{poly}(n))$ on a quantum computer and time $\mathcal{O}(|S_k|)$ on a classical computer. Since for random Erdös-Renyi complexes, $|S_k| \approx n^{k/2}$, this gives an exponential quantum advantage for the \emph{estimation} step if $k=\Theta(n).$

Similarly, an exponential quantum advantage may appear if we go beyond clique complexes to abstract simplicial complexes. Abstract simplicial complexes are not defined via the cliques of an underlying graph but via a collection of inclusion-maximal simplices that are closed under subset-taking. In general, this collection might be exponentially large and reintroduce the data-loading problem. However, in some typical real world examples such a description is compact. There are many well-known examples of abstract simplicial complexes with exponentially large Betti numbers. One example is the \emph{$k$-skeleton} of the $n$-simplex, which is the simplicial complex on $n+1$ vertices containing all possible simplices of dimension less than or equal to $k.$ These simplicial complexes have close to maximal Betti numbers: The $k$-th Betti number of the $k$-skeleton of an $n$-simplex equals $\beta_k = {n\choose{k+1}}.$ They are also compactly specified by a list of maximal non-simplices. For these and similar abstract simplicial complexes with exponentially large Betti numbers, the LGZ-algorithm solves the \emph{estimation} step efficiently. 
\subsection{Oracle Construction}
Solving the \emph{estimation} step efficiently enables the efficient computation of Betti numbers only if we have access to random sampling of $k$-simplices in the complex.  While random sampling of simplices is in general a hard problem,
we note that datasets that permit random simplicial sampling could be available in real-world applications.   

For example, consider the problem of calculating the homology of Facebook, where $k+1$ Facebook users form a simplex if they are members of the same Facebook group (the same $k+1$ users can all be members of more than one group). The resulting simplicial complex is not a clique complex. In this setting, random sampling of simplices can be achieved efficiently by simply querying Facebook for a random element of their stored list of groups, assuming classical access to this list. If we moreover assume coherent \footnote{In general, coherent access might require a large QRAM} access to a membership function $f_k$ (see Section II. C), a quantum algorithm can efficiently solve the \emph{estimation} step by following the steps of the LGZ-algorithm: Randomly sample a $k$-simplex from the complex, apply quantum phase estimation to it with the exponentiated simplicial Laplacian and record the relative frequency of the zero eigenvalue using quantum counting. In this model, a quantum algorithm can thus compute the simplicial homology of Facebook efficiently, provided the Betti numbers of Facebook are comparable to the number of simplices. At the same time, even with access to efficient simplex sampling, the best classical approach seems to require the computation of the nullity of a exponentially large matrix, which takes exponential time \cite{ubaru2016fast}.  

In general, if we are given a list of sets and their members, and define a $k$-simplex to be $k+1$ points all of which lie in the same set, this describes an abstract simplicial complex from which we can efficiently sample simplices. If we further assume coherent access to a membership function for the simplices, a modification of the LGZ-algorithm supplies an exponential advantage over classical algorithms for computing homology when the $k$-th Betti number is large and comparable to the number of $k$-simplices, and the spectral gap is scaling inverse-polynomially.
\section{Complexity of Estimating Normalized Betti Numbers}
Motivated by practical applications of TDA, the previous chapters have focused on computing Betti numbers exactly or up to multiplicative error. From a quantum algorithms perspective, a more natural output is an additive approximations to the \emph{normalized} $k$-th Betti number \begin{equation}
    c_k = \frac{\beta_k}{|S_k|},
\end{equation} i.e. to the $k$-th Betti number divided by the number of $k$-simplices in the simplicial complex. Indeed, this is the original output of the LGZ-algorithm \cite{LloydAlgo}, although there mistakenly labeled a multiplicative approximation. Later work \cite{reviewBetti,GoogleBettiBerry} focuses on multiplicative approximations. As we discussed in sections 2 and 5, in almost all cases where $c_k$ can be estimated in polynomial time, the normalized Betti number is exponentially small. Any inverse polynomial additive error $\varepsilon$ can therefore not distinguish it from zero, in agreement with the hardness of \textbf{Homology}. 

Nevertheless, the complexity of estimating normalized Betti numbers (BNE) is a highly interesting topic offering a potentially exponential quantum advantage. In a recent paper \cite{quantumAdvantage}, Dunjko et al. proved  that a generalization of this task, called low-lying spectral density estimation (LLSD), is DQC1-hard. In another recent work, Cade et al. \cite{cade2021complexity} proved that estimating normalized Betti numbers for arbitrary chain complexes (as opposed to clique complexes) is also DQC1-hard. The complexity class DQC1 (or deterministic quantum computation with one clean qubit) is the class of decision problems solvable by a one clean qubit machine in polynomial time with low error probability (e.g. at most 1/3 for all instances). A \emph{one clean qubit machine} is a $n$-qubit system, where the first qubit starts out in the pure state $\ket{0}$, while all other $n-1$ qubits start out in the maximally mixed state. The initial density matrix of the one clean qubit machine is thus \begin{equation} \rho_{\text{DQC1}}= \ket{0}\bra{0} \tens \frac{1}{2^{n-1}} \mathbb{1}_{n-1},\end{equation} where $\mathbb{1}_{n-1}$ is the identity matrix on $n-1$ qubits. We are then allowed to apply any polynomial-sized quantum circuit to $\rho$ and subsequently measure the first qubit.

The results of \cite{quantumAdvantage,cade2021complexity} imply a relationship between the problem of estimating normalized Betti numbers and other well-known DQC1-complete problems, such as the \emph{power-of-one-qubit} task of estimating normalized traces of unitaries \cite{Knill_1998} and the approximation of the Jones polynomial at a fifth root of unity for the trace closure, due to Shor and Jordan  \cite{Pshor}. Theorem \ref{thm:Betti} and Theorem \ref{thm:Homology} further show that both the Betti number problem and the Jones polynomial problem are \#P-hard in the exact case and that additive approximations to both problems are not value-distinguishing \cite{kuperberg2015hard}.
It is widely believed that the one clean qubit model can not be simulated by a classical computer in polynomial time. 

How do the DQC1-completeness results of \cite{quantumAdvantage,cade2021complexity} relate to Theorem \ref{thm:Betti} and Theorem \ref{thm:Homology}? As emphasized by the respective authors, it is not clear whether either result actually applies to quantum algorithms for TDA, as they both concern (possibly strong) generalizations of the central task of TDA, which is estimating Betti numbers of Vietoris-Rips complexes.  
Let us first discuss the potential complexity-theoretic gap between BNE and LLSD, as defined in \cite{quantumAdvantage}, Section 3A. BNE solves LLSD only if the Hermitian matrix that is used as input for LLSD is the simplicial Laplacian $\Delta_k$ of a simplicial complex. Whether LLSD remains DQC1-hard if restricted to simplicial Laplacians remains an open problem. Similarly, the results of \cite{cade2021complexity} apply to more general chain complexes, and it is not clear whether the DQC1-hardness remains if the input is restricted to clique complexes. For arbitrary chain complexes, the Betti number estimation problem is trivially NP-hard.


If BNE were shown to be DQC1-hard, this would greatly strengthen the claim that the version of the LGZ-algorithm targeting an additive approximation to the normalized Betti number provides an exponential quantum advantage. Strengthening the above hardness results is thus an exciting direction for further research. Nevertheless, it would not contradict Theorem \ref{thm:Betti} and Theorem \ref{thm:Homology}, since a computational advantage for additively estimating normalized Betti numbers does not necessarily translate to a computational advantage for multiplicatively estimating Betti numbers. 
\section{Discussion}
In this paper, we investigated the computational complexity of computing Betti numbers in topological data analysis. We proved that the problem of computing Betti numbers of clique complexes exactly is \#P-hard, while the problem of estimating them up to any multiplicative error is NP-hard. Moreover, we show that both problems retain their hardness if the input is restricted to clique-dense complexes, a regime of measure zero in which the LGZ-algorithm performs best. Because quantum computers are not expected to be able to solve \#P-hard or NP-hard problems in subexponential time, our results imply that quantum algorithms for TDA run in exponential time in the worst case. Contrary to previous claims, this holds even in the optimal regime of clique-dense complexes. We verify our results through a detailed analysis of random Vietoris-Rips complexes, from which we conclude that the LGZ-algorithm achieves a quadratic speedup over the best classical algorithm for asymptotically almost all inputs. 

Our results show that the bottleneck of quantum algorithms for TDA is the construction of a clique complex from a description of its underlying graph. If the simplicial complex is instead specified in a way that allows random sampling from the simplices in the complex, a modification of quantum algorithms for TDA can compute Betti numbers efficiently if certain conditions are met. Under the same conditions, classical algorithms for estimating Betti numbers still require a runtime exponential in the number of data points. In this case, an exponential quantum advantage is recovered.  We note that data composed of lists of sets and their members takes the form required to retain an exponential quantum advantage, and provide examples where such input models appear naturally. 

\begin{acknowledgments}
A. Schmidhuber acknowledges helpful discussions with Ryan Babbush and Marcos Crichigno. S. Lloyd acknowledges helpful discussions with Michele Reilly and was funded by DARPA and the DOE. 
\end{acknowledgments}

\appendix

\bibliography{references}

\end{document}